\documentclass[pra,aps,twocolumn,twoside,superscriptaddress,longbibliography,nofootinbib]{revtex4-1}
\usepackage{times}

\usepackage{graphicx}
\usepackage[pdftex,colorlinks=true,linkcolor=blue,citecolor=blue,urlcolor=black]{hyperref}
\usepackage{amsmath, amsthm, amssymb}
\usepackage{subfigure}
\usepackage{comment}
\usepackage{url}
\usepackage[ruled,lined,linesnumbered]{algorithm2e}
\usepackage{tikz}

\newcommand{\R}{\mathbb{R}}

\newcommand{\Z}{\mathbb{Z}}

\newcommand{\E}{\mathbb{E}}

\DeclareMathOperator{\poly}{poly}

\newcommand{\be}{\begin{equation}}
\newcommand{\ee}{\end{equation}}
\newcommand{\bea}{\begin{eqnarray}}
\newcommand{\eea}{\end{eqnarray}}
\newcommand{\bes}{\begin{equation*}}
\newcommand{\ees}{\end{equation*}}
\newcommand{\beas}{\begin{eqnarray*}}
\newcommand{\eeas}{\end{eqnarray*}}

\makeatletter
\newtheorem*{rep@theorem}{\rep@title}
\newcommand{\newreptheorem}[2]{%
\newenvironment{rep#1}[1]{%
 \def\rep@title{#2 \ref{##1} (restated)}%
 \begin{rep@theorem}}%
 {\end{rep@theorem}}}
\makeatother

\newcommand{\boxalgm}[3]{
\renewcommand{\figurename}{Algorithm}
\begin{figure}[tb]
\begin{center}
\noindent \framebox{
\begin{minipage}{.45\textwidth}
#3
\end{minipage}
}
\caption{#2}
\label{#1}
\end{center}
\end{figure}
\renewcommand{\figurename}{Figure}
}

\newtheorem{thm}{Theorem}
\newtheorem*{thm*}{Theorem}

\newtheorem{lem}[thm]{Lemma}
\newtheorem*{lem*}{Lemma}

\newreptheorem{thm}{Theorem}
\newreptheorem{lem}{Lemma}
\newreptheorem{cor}{Corollary}

\begin{document}

\title{Quantum speedup of branch-and-bound algorithms}

\author{Ashley Montanaro}
\affiliation{School of Mathematics, University of Bristol, Bristol, UK.}
\email{ashley.montanaro@bristol.ac.uk}

\begin{abstract}
Branch-and-bound is a widely used technique for solving combinatorial optimisation problems where one has access to two procedures: a branching procedure that splits a set of potential solutions into subsets, and a cost procedure that determines a lower bound on the cost of any solution in a given subset. Here we describe a quantum algorithm that can accelerate classical branch-and-bound algorithms near-quadratically in a very general setting. We show that the quantum algorithm can find exact ground states for most instances of the Sherrington-Kirkpatrick model in time $O(2^{0.226n})$, which is substantially more efficient than Grover's algorithm.
\end{abstract}

\maketitle



Quantum computers can solve certain problems, such as simulation of quantum-mechanical systems~\cite{georgescu14} and integer factorisation~\cite{shor97}, exponentially faster than the best classical algorithms known. As well as these special-purpose algorithms, there are general-purpose quantum algorithms which can outperform their classical counterparts more modestly for a wide range of problems within the domains of constraint satisfaction and optimisation~\cite{montanaro16a}. A famous example is Grover's algorithm for unstructured search~\cite{grover97}, which can be applied to find the minimal value in a set of size $N$ with $O(\sqrt{N})$ evaluations of values in the set~\cite{durr96}. This algorithm achieves a quadratic speedup over exhaustive classical search. However, for many problems encountered in practice, there are more efficient classical algorithms than exhaustive search, which take advantage of the structure of the problem. This can apply to NP-complete problems, which are expected not to have polynomial-time algorithms, yet which can sometimes be solved surprisingly efficiently in practice.

One of the most successful general approaches to solving constrained optimisation problems is known as {\em branch-and-bound}. This approach can be applied to problems where the goal is to find a minimal-cost valid solution, in a setting where one has access to two functions: a bounding function {\tt cost} that, for a given subset of the set of possible solutions, returns a lower bound on the cost of any valid solution in that subset; and a branching rule {\tt branch} to be applied if a subset of possible solutions cannot yet be ruled out, which will divide that subset into two or more ``live'' subsets to be explored in later iterations. Then the branch-and-bound approach explores subsets of potential solutions, ruling out those where the cost of any valid solution is too high (e.g.\ higher than the lowest cost of a valid solution found so far). The goal is to use this additional information to avoid exploring every possible solution. Clearly, this approach can equivalently be applied to problems where the goal is to maximise the value of a valid solution.

Although the number of subsets produced by the branching steps can grow exponentially, implying an exponential running time, algorithms based on this technique can sometimes find exact solutions to instances of hard optimisation problems substantially beyond the reach of unstructured search. For example, integer linear programming problems can be solved using a branch-and-bound approach where the {\tt cost} function is based on relaxing to linear programming problems (see Appendix \ref{app:ilp} for more details).

Here we describe a general quantum approach to accelerate classical branch-and-bound algorithms almost quadratically. The quantum branch-and-bound algorithm can be applied to speed up any classical algorithm that fits into the branch-and-bound paradigm (formally defined below). These include algorithms for integer linear programming; nonlinear programming; the travelling salesman problem; and more~\cite{lawler66}.

The quantum algorithm is based on quantum subroutines that speed up a related class of classical algorithms: backtracking algorithms~\cite{montanaro18,ambainis17}. Backtracking is an approach that solves constraint satisfaction problems given the ability to determine whether a partial solution to the problem could be extended to a full solution. Backtracking algorithms can be interpreted as the special case of branch-and-bound algorithms where the {\tt cost} function either returns 0 (for a valid solution to the problem) or $\infty$ (for an invalid potential solution, or a partial solution that cannot be extended to a full solution), so the algorithm described here can be seen as generalising the results of~\cite{montanaro18,ambainis17}.

This result contrasts with other quantum approaches to solve hard optimisation problems, such as the adiabatic algorithm~\cite{farhi00} and the quantum approximate optimisation algorithm~\cite{farhi14}, in that the branch-and-bound algorithm guarantees to find the minimal-cost solution with arbitrarily high probability; however, for certain problems its running time can be long (e.g.\ exponential in the input size).

One area where branch-and-bound algorithms have been successfully applied classically is finding ground states of spin models~\cite{kobe78,hartwig80,palassini03,packebusch16}. For example, branch-and-bound has been used to find the largest exact ground states known of instances of the Bernasconi model~\cite{bernasconi87}, corresponding to binary sequences with minimal autocorrelation. The fastest known branch-and-bound algorithms for this model have runtime estimated numerically as approximately $O(2^{0.79n})$~\cite{packebusch16}; the quantum branch-and-bound algorithm would improve this scaling to approximately $O(2^{0.4n})$.

Another spin model addressed using branch-and-bound is the well-studied Sherrington-Kirkpatrick (S-K) model~\cite{sherrington75}, which is the family of classical Hamiltonians $H(x) = \sum_{1 \le i < j \le n} a_{ij} x_i x_j$ where $x \in \{\pm 1\}^n$ and $a_{ij}$ are distributed according to the normal distribution $N(0,1)$. Finding the lowest-energy state for such a Hamiltonian can be achieved in time $O(2^{n/2} \poly(n))$ using Grover's algorithm. Recently, Callison et al.~\cite{callison19} have described an intriguing quantum algorithm that uses quantum walks to solve the S-K model, and gave evidence based on numerical experiments for small $n$ that the runtime of the algorithm should be approximately $O(2^{0.41n})$. Here we apply the quantum branch-and-bound algorithm to speed up a simple classical branch-and-bound algorithm~\cite{hartwig80} for the S-K model, and show rigorously that the runtime of the quantum algorithm is $O(2^{0.226n})$ on most instances of size $n$, which is substantially more efficient than Grover search. Numerical evidence suggests that the runtime in practice could be as low as $O(2^{0.186n})$.

Branch-and-cut methods have also been used to find ground states of spin glasses~\cite{palassini03}; these could be accelerated using the same quantum approach.


{\bf Model for branch-and-bound algorithms.} For a problem to be accessible to branch-and-bound, we must have access to {\tt cost} and {\tt branch} procedures. Each takes as input a subset $S$ of potential solutions to an optimisation problem, perhaps chosen from a restricted family of subsets. {\tt cost}$(S)$ returns a lower bound on the cost of any solution within $S$. {\tt branch}$(S)$ either returns that $S$ only contains one element, or splits $S$ into two or more disjoint sets $S_1,\dots,S_k$. For simplicity, here we assume that $k$ is a fixed constant. For the behaviour of {\tt cost} to be reasonable, we must have {\tt cost}$(S) \ge $ {\tt cost}$(S')$ whenever $S \subseteq S'$. We assume that for all subsets $S$, either {\tt cost}$(S) \in [0,c_{\max}]$ for some known $c_{\max}$, or {\tt cost}$(S) = \infty$, where the latter corresponds to $S$ containing no valid solutions. We can assume essentially without loss of generality that {\tt cost} is integer-valued; real-valued cost functions can be handled by truncating their output to precision $\delta$, and multiplying all costs by $1/\delta$. In this situation $c_{\max}$ effectively acts as a precision parameter.

An abstraction of the problem of finding a minimal-cost solution to a problem given access to the {\tt cost} and {\tt branch} procedures is the following model of search within trees~\cite{karp86}, illustrated in Figure \ref{fig:tree}. The search space is described by a rooted tree where each node $v$ is either labelled with an integer $c(v)$ between 0 and $c_{\max}$ or with $\infty$, and satisfies the promise that if $w$ is a child of $v$, then $c(w) \ge c(v)$. We are given query access to two oracles, each of which takes as input a node $v$. One oracle returns $c(v)$ and the other returns the children of $v$, if there are any. The goal is to find a leaf node $v$ such that $c(v)$ is minimised, while making the minimal number of queries.

In this abstraction, a node in the tree represents a subset of possible solutions (with a leaf representing a single possible solution), and its label represents a lower bound on the cost of any solution in that subset. A label of $\infty$ represents that there is no valid solution in that subset. Revealing the children of a node corresponds to splitting a set of potential solutions into subsets. Note that the backtracking approach for solving constraint satisfaction problems corresponds to exploring a tree of this form where each node is labelled either with 0 or $\infty$.

The best classical strategy for solving the search problem in such a tree is to maintain a list of ``live'' nodes (those whose children have not yet been explored), and always to choose to explore the live node with the lowest cost~\cite{karp86}. This strategy is known as best-first search. Although it is optimal in terms of query complexity, it may require a very large amount of space to store the list of live nodes, so in practice other strategies may be preferred (such as depth-first search); strategies which achieve near-optimal query complexity in limited space are known~\cite{karp86}.

We will consider particular subtrees of the overall tree, obtained by truncating it at a particular cost $c$, i.e.\ deleting all nodes whose labels are greater than $c$. This is equivalent to changing the {\tt cost} function to a {\tt cost'} function such that {\tt cost'}$(x) = \infty$ if {\tt cost}$(x) > c$. This transformation clearly preserves the tree structure and the monotonicity of the {\tt cost} function. We also have that truncating the tree at a cost $c_{\min}$ equal to the minimal cost of a valid solution preserves the presence of a minimal-cost valid solution in the tree.

As observed by Karp and Zhang~\cite{karp93}, this kind of truncation controls the complexity of a class of classical search algorithms: Any algorithm which outputs all the minimal-cost leaves must explore the entire tree truncated at cost $c_{\min}$. Otherwise, it could not be sure that it had found all the minimal-cost leaves. In particular, if all solutions have distinct costs, any algorithm which outputs the minimal-cost solution must explore the whole tree truncated at that cost.

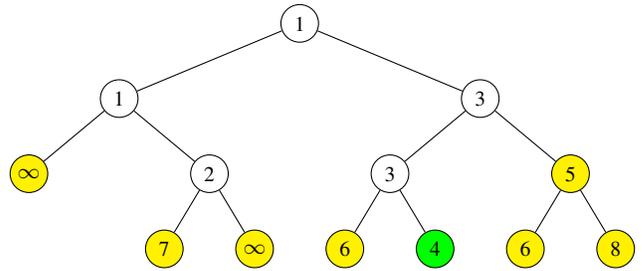
\begin{figure}
  \begin{tikzpicture}[xscale=1.2,every node/.style={font=\footnotesize,circle,fill=white,draw,minimum size=0.5cm,inner sep=0.4mm}]
  \node (1) at (0,0) {1};
  \node (20) at (-2,-1) {1};  \node (21) at (2,-1) {3};
  \node[fill=yellow] (300) at (-3,-2) {$\infty$};  \node (301) at (-1,-2) {2};  \node (310) at (1,-2) {3};  \node[fill=yellow] (311) at (3,-2) {5};
 \node[fill=yellow] (4010) at (-1.5,-3) {7};   \node[fill=yellow] (4011) at (-0.5,-3) {$\infty$};   \node[fill=yellow] (4100) at (0.5,-3) {6};    \node[fill=green]  (4101) at (1.5,-3) {4};    \node[fill=yellow] (4110) at (2.5,-3) {6};    \node[fill=yellow]  (4111) at (3.5,-3) {8}; 
  \draw (1) -- (20); \draw (1) -- (21); 
  \draw (20) -- (300); \draw (20) -- (301); \draw (21) -- (310); \draw (21) -- (311);
 \draw (301) -- (4010); \draw (301) -- (4011);  \draw (310) -- (4100); \draw (310) -- (4101);    \draw (311) -- (4110); \draw (311) -- (4111);  
   \end{tikzpicture}
\caption{An example tree corresponding to a branch-and-bound algorithm. Nodes are labelled with their cost bounds, which are non-decreasing on any path from the root to a leaf. The green node is the optimal solution; yellow nodes are removed if the tree is truncated at the optimal cost.}
\label{fig:tree}
\end{figure}


{\bf Statement of results.} We can now state our main result. Let $\mathcal{T}$ be the tree corresponding to a branch-and-bound algorithm $\mathcal{A}$. Let $d$ be the depth of $\mathcal{T}$, let $c_{\min}$ be the minimal cost of a valid solution, and let $T_{\min}$ be the size of the truncated tree with cost bound $c_{\min}$. (If there is no solution, $c_{\min} = \infty$, and $T_{\min}$ is the size of the whole tree $\mathcal{T}$.) Fix a constant $\epsilon > 0$. Then there is a quantum algorithm which uses
\[ \widetilde{O}\left( \sqrt{T_{\min}} d^{3/2} \log c_{\max} \right)  \]
calls to {\tt cost} and {\tt branch}, and except with failure probability at most $\epsilon$, returns a solution with minimal cost, if one exists, and otherwise returns ``no solution''. The $\widetilde{O}$ notation hides polylogarithmic factors in $d$, $1/\epsilon$ and $\log c_{\max}$.

Assuming that $T_{\min} \gg \poly(d)$, $T_{\min} \gg \log c_{\max}$, this is roughly a quadratic speedup over any possible classical branch-and-bound search algorithm which finds {\em all} minimal-cost solutions in the tree corresponding to $\mathcal{A}$, whose complexity (as discussed above) is lower-bounded by $T_{\min}$. In the usual case where there is a unique minimal-cost solution, the speedup is approximately quadratic over any possible classical branch-and-bound algorithm (i.e.\ one that uses only the {\tt cost} and {\tt branch} procedures).


{\bf Quantum branch-and-bound algorithm.} In order to state the quantum branch-and-bound algorithm, we will need two quantum-algorithmic ingredients. Both relate to determining properties of trees, in a model where only the root node is known in advance, and the structure of the tree can only be revealed via ``oracle'' queries to nodes. A query to a node reveals the neighbours of the node and whether the node is marked. This model disallows the straightforward use of methods such as Grover's algorithm to search in the tree.

The first ingredient is {\em quantum tree search}. It follows from~\cite{montanaro18,belovs13,belovs13a} that there is a quantum algorithm which, given $\epsilon$ and oracle access to a tree with depth at most $d$, $T$ nodes and maximal degree $k = O(1)$, makes
\be \label{eq:treesearch} O\left(\sqrt{T} d^{3/2} \log d \log(1/\epsilon) \right) \ee
queries and performs $O(1)$ other elementary operations per query, and except with failure probability at most $\epsilon$:
\begin{itemize}
\item If the tree contains at least one marked node, the algorithm returns the label of a marked node;
\item If the tree does not contain any marked nodes, the algorithm returns ``not found''.
\end{itemize}

Jarret and Wan have described an algorithm~\cite{jarret18} which solves the same tree search problem with complexity bounded by $O(\sqrt{Td} \log^4(md)\log(m/\epsilon) )$, where $m$ is the number of marked nodes. This is an improvement on (\ref{eq:treesearch}) by up to a factor of almost $d$ when $m$ is small. However, if $m$ is very large (e.g.\ exponential in $n$), this complexity bound can be larger than (\ref{eq:treesearch}). Here, we will need to apply quantum tree search in a setting where we have no upper bound on $m$, which could be as large as $T$, which in turn can be exponentially large in $n$. So we will state complexity bounds based on (\ref{eq:treesearch}), though in some cases (e.g.\ when $T$ is small) the algorithm of Jarret and Wan may be more efficient.

When we use the quantum tree search algorithm we will have access to an upper bound on the size of the tree. Although, as stated above, having access to such an upper bound is not necessary, it can be seen by inspecting the proof of correctness in~\cite{montanaro18} that the algorithm is simplified somewhat given this additional information. However, this does not affect its asymptotic complexity.

The second ingredient we will need is {\em quantum tree size estimation}~\cite{ambainis17}. The quantum tree size estimation algorithm, given query access to a tree with depth at most $d$, $T$ nodes and maximal degree $k = O(1)$, and parameters $T_0$, $\epsilon$, $\delta$, makes
\[ O\left(\frac{\sqrt{T_0d}}{\delta^{3/2}} \log^2 (1/\epsilon) \right) \]
queries and performs $O(\log T_0)$ other elementary operations per query, and except with failure probability at most $\epsilon$:
\begin{itemize}
\item If $T \le T_0 / (1+\delta)$, the algorithm outputs $\widetilde{T} \le T_0$ such that $|\widetilde{T}-T| \le \delta T$;
\item If $T > (1+\delta) T_0$, the algorithm outputs ``contains more than $T_0$ nodes''.
\end{itemize}

The restriction to $k = O(1)$ in these algorithms is not significant, as any node with degree $k$ can be replaced with a binary tree of depth $O(\log k)$. Hence, if each node had degree at most $k$, the tree size $T$ would increase by a constant factor and the depth $d$ would increase by an $O(\log k)$ factor, corresponding to the complexity bounds increasing by a factor polylogarithmic in $d$.

Let Count$_c(T_0,\epsilon,\delta)$ denote the quantum tree size estimation algorithm applied to count the number of nodes in the truncated tree corresponding to the backtracking algorithm with cost bound $c$, where $d$ and $k$ are fixed; and similarly let Search$_c(\epsilon)$ denote the quantum tree search algorithm applied to search in the truncated tree with cost bound $c$, where a node is marked if it corresponds to a valid solution to the original problem (i.e.\ was a leaf with finite cost in the original tree). Assume that we have an upper bound $c_{\max}$ such that all valid solutions have cost strictly less than $c_{\max}$, and an upper bound $T_{\max} \le k^d$ on the size of the tree; further assume for simplicity that $c_{\max}$ is a power of 2.

\boxalgm{alg:qbb}{Quantum branch-and-bound algorithm with failure parameter $\epsilon$.}{
\begin{enumerate}
\item Set $T \leftarrow 1$, $c_{\text{old}} \leftarrow 0$, $\epsilon' \leftarrow \epsilon / (K d \log_2 c_{\max})$ for some sufficiently large constant $K$.
\item While $T \le T_{\max}$:
\begin{enumerate}
\item If $T > T_{\max}/2$, $c_{\text{new}} \leftarrow c_{\max}$. Otherwise, set $c_{\text{new}} \leftarrow 0$ and for $i=1$ to $\log_2 c_{\max}$:
\begin{enumerate}
\item If Count$_{c_{\text{new}} + c_{\max}/2^i}(T,\epsilon',1/2)$ does not return ``contains more than $T$ nodes'', set $c_{\text{new}} \leftarrow c_{\text{new}} + c_{\max}/2^i$.
\end{enumerate}
\item Run Search$_{c_{\text{new}}}(\epsilon')$. If it returns the label of a solution, use binary search on $c$ between $c_{\text{old}}$ and $c_{\text{new}}$ within Search$_c(\epsilon')$ to find the minimal $c$ such that a solution with cost $c$ exists, and return that solution.
\item Set $T \leftarrow 2T$, $c_{\text{old}} \leftarrow c_{\text{new}}$.
\end{enumerate}
\item Return ``no solution''
\end{enumerate}
}

Then the quantum branch-and-bound algorithm is stated formally as Algorithm \ref{alg:qbb}. The intuitive idea behind it is as follows: we want to find a cost $c \ge c_{\min}$ such that the size $T_c$ of the tree truncated at cost $c$ is not much greater than $T_{\min}$. Given such a $c$, we can find a minimal-cost solution within the tree using Search$_c$, with a complexity of $O(\sqrt{T_{\min}} \poly(d))$ queries. And to find such a $c$ efficiently, we can perform a binary search on $c$ using Count$_c$, to find the maximal $c$ such that the tree size is smaller than some upper bound, which is not much greater than $T_{\min}$. There is the technicality that Count$_c(T,\epsilon,\delta)$ may not return the correct answer for some choices of $c$ and $T$. However, this turns out not to affect the correctness or performance of the binary search, because Count$_c(T,\epsilon,\delta)$ does always return the correct answer when the size of the tree truncated at $c$ is sufficiently small with respect to $T$. The formal proof of correctness and runtime of the algorithm is deferred to Appendix \ref{app:qbb}.


{\bf Sherrington-Kirkpatrick spin glass.} The Sherrington-Kirkpatrick (S-K) model~\cite{sherrington75} is the family of classical Hamiltonians $H(x) = \sum_{1 \le i < j \le n} a_{ij} x_i x_j$ where $x \in \{\pm 1\}^n$ and $a_{ij}$ are distributed according to the normal distribution $N(0,1)$. Let $A = (a_{ij})$ denote the corresponding square matrix where $a_{ij} = 0$ for $i \ge j$. Given a Hamiltonian $H$ described by the matrix $A$, our computational task is to determine $E_{\min}(H) := \min_x H(x)$.

Determination of the limiting form of the expected ground state energy $E_{\min} := \E_H[E_{\min}(H)]$ as $n \rightarrow \infty$ has been a question of extensive interest within the theory of spin glasses. A precise limiting expression is known for this quantity, which evaluates to $E_{\min} = (-0.763167\dots + o(1))n^{3/2}$. This formula was conjectured by Parisi~\cite{parisi80} and later proven correct by Talagrand~\cite{talagrand06}. Although an explicit expression, evaluating it numerically is non-trivial; however, the constant factor is now known to many digits of precision~\cite{crisanti02,schmidt08}.

Finding the ground state energy of general Ising model Hamiltonians (of the form of $H$ with arbitrary coefficients $a_{ij}$) is NP-hard, which holds even given locality restrictions on the pairs $i$, $j$ such that $a_{ij} \neq 0$~\cite{barahona82}. The S-K model is a natural family of Ising model Hamiltonians, and has been a target for a number of different algorithmic approaches, both heuristic (e.g.~\cite{boettcher05,pelikan08}) and exact~\cite{kobe78,hartwig80,kobe03}. Although it was recently proven that ground state energies of the S-K model can be approximated efficiently~\cite{montanari18}, there is no known efficient (i.e.\ polynomial-time in $n$) method to compute them exactly.

There is a straightforward exact approach to computing $E_{\min}(H)$ which fits into the branch-and-bound paradigm and was proposed in~\cite{hartwig80}. Variables $x_1,\dots,x_n \in \{\pm1\}$ are assigned values sequentially. To determine a lower bound on the cost of any assignment beginning with a partial assignment $x_1,\dots,x_\ell$, we observe that
\begin{multline*}
H(x_1,\dots,x_\ell,z_{\ell+1},\dots,z_n) =\\ \sum_{1 \le i < j \le \ell} a_{ij} x_i x_j + \sum_{i=1}^{\ell} \sum_{j =\ell+1}^n a_{ij} x_i z_j +\!\!\!\! \sum_{\ell+1 \le i < j \le n} \!\!\!\! a_{ij} z_i z_j
\end{multline*}
and hence
\beas
&& \min_z H(x_1,\dots,x_\ell,z_{\ell+1},\dots,z_n)\\
&=& \!\!\!\!\sum_{1 \le i < j \le \ell} \!\!\!\!a_{ij} x_i x_j + \min_z \left( \sum_{i=1}^{\ell} \sum_{j =\ell+1}^n a_{ij} x_i z_j +\!\!\!\!\!\! \sum_{\ell+1 \le i < j \le n}\!\!\!\!\!\! a_{ij} z_i z_j \right)\\
\label{eq:skbound} &\ge& \!\!\!\! \sum_{1 \le i < j \le \ell} \!\!\!\! a_{ij} x_i x_j - \sum_{j =\ell+1}^n \left| \sum_{i=1}^{\ell} a_{ij} x_i \right| + \min_z \!\!\!\!\sum_{\ell+1 \le i < j \le n} \!\!\!\!a_{ij} z_i z_j\\
&=:& \text{Bound}_A(x).
\eeas
The first and second components of Bound$_A(x)$ can be computed from $x$ in time $O(n^2)$. The third component is equal to $\min_x H'(x)$, where $H'$ is formed from $H$ by deleting the first $\ell$ rows and columns of the matrix $A=(a_{ij})$. Thus, if we have already solved the $n-1$ minimisation problems obtained by restricting $A$ to the corresponding submatrices, we can compute Bound$_A(x)$ efficiently. As the runtime of the overall algorithm is exponential in $n$, this additional cost does not significantly affect its overall complexity.

This algorithm fits into the formal model for branch-and-bound discussed above in a straightforward way: the branch-and-bound tree is a binary tree, where nodes at depth $\ell$ correspond to $\ell$-bit strings, for $0 \le \ell \le n$, and the {\tt cost} function is Bound$_A$. So the quantum branch-and-bound algorithm can immediately be applied to accelerate it quadratically. Note that {\tt cost} should be non-negative and integer-valued; this can be achieved by truncating the real values $a_{ij}$ after $p$ binary digits to produce a new Hamiltonian $\widetilde{H}$, rescaling by $2^p$, and performing an overall energy shift. It is sufficient to take $p = O(\log n)$ to achieve $E_{\min}(\widetilde{H}) = E_{\min}(H) + o(1)$. This corresponds to a cost bound for the integer-valued problem $c_{\max} = \poly(n)$.

Let $T_A$ denote the size of the branch-and-bound tree corresponding to the Hamiltonian $H$ with matrix $A$, using cost function Bound$_A$, truncated optimally at cost $E_{\min}(H)$. In Appendix \ref{app:sktree} we prove that for sufficiently large $n$, $\Pr_A[T_A \ge 2^{0.451n}] \le 0.01$. So for the vast majority of instances of size $n$, the runtime of the classical branch-and-bound algorithm is $O(2^{0.451n})$ steps, which is already faster than Grover's algorithm. The runtime of the quantum branch-and-bound algorithm on these instances is $O( 2^{0.226n} )$ steps, observing that the depth $d$ of the branch-and-bound tree is equal to $n$, and $\log c_{\max} = O(\log n)$, so these only contribute lower-order terms to the complexity. The constant $0.01$ could be made arbitrarily small.

This rigorous result is only an upper bound on the tree size of the classical algorithm (and by extension the runtime of the quantum algorithm), which may not be tight. To investigate this, a depth-first variant of the classical algorithm was implemented and run on instances of the S-K model for $n \le 50$. Extrapolation of the tree sizes obtained suggests that the true scaling of $T_A$ for random $A$ is $O(2^{0.371n})$, which would correspond to a quantum runtime of $O(2^{0.186n})$. See Appendix \ref{app:numerics} for further details.


{\bf Conclusions and further work.} We have described a quantum algorithm which can accelerate classical branch-and-bound algorithms in a very general setting. We finish by discussing potential routes to improving the results of this work. First, in some cases the square-root dependence of the runtime on $T_{\min}$ cannot be improved; even in the special case where the {\tt cost} function either evaluates to 0 or $\infty$, there are depth-$d$ trees for which the quantum algorithm requires $\Omega(\sqrt{T_{\min} d})$ queries to determine whether there is a leaf of cost 0~\cite{aaronson05}. However, it might be possible to improve the depth dependence by a factor of up to $d$, e.g.\ by extending ideas of~\cite{jarret18}.

Ambainis and Kokainis~\cite{ambainis17} have given a quantum algorithm which, given a deterministic classical algorithm that explores a search tree to find a marked node and uses $Q$ queries to do so, can find a marked node using $\widetilde{O}(\sqrt{Q}d^{3/2})$ queries. This improves on the quantum tree search algorithm used here in that its complexity depends on $Q$, rather than the size $T$ of the whole tree. It seems unclear whether this algorithm could be applied to accelerate Algorithm \ref{alg:qbb} directly, for at least two reasons. First, the algorithm as designed also uses the Count subroutine, whose runtime depends on the entire tree size. Second, in practice any classical or quantum branch-and-bound algorithm is likely to explore the whole tree of potential solutions, in order to ensure that it has not omitted any solutions with lower cost than the current best solution.

Finally, in order to determine whether the quantum branch-and-bound algorithm will genuinely outperform the best classical methods for problems of practical interest, once all overheads are taken into account, a more detailed analysis of the algorithm's runtime should be undertaken, extending previous analysis for backtracking~\cite{campbell18}.


{\bf Acknowledgements.} I acknowledge support from the QuantERA ERA-NET Cofund in Quantum Technologies implemented within the European Union's Horizon 2020 Programme (QuantAlgo project), EPSRC Early Career Fellowship EP/L021005/1 and EPSRC/InnovateUK grant EP/R020426/1. This project has received funding from the European Research Council (ERC) under the European Union's Horizon 2020 research and innovation programme (grant agreement No.\ 817581). I would like to thank Viv Kendon for insight into the results of~\cite{callison19}, and Stephen Piddock for helpful discussions throughout.


\bibliographystyle{plain}
\bibliography{../../thesis}


\appendix


\section{Example: Integer Linear Programming}
\label{app:ilp}

To gain some intuition for how the results presented here could be applied, in this appendix we describe one simple and well-known application of branch-and-bound techniques: integer linear programming. An integer linear program (ILP) is a problem of the form:
\begin{align*}
\text{minimise } & c^T x\\
\text{subject to } & Ax \ge b,\\
& x \ge 0,\\
& x \in \Z^n
\end{align*}
where $b$ and $c$ are vectors, $A$ is a matrix, and inequalities are interpreted componentwise. Integer linear programming problems have many applications, including production planning, scheduling, capital budgeting and depot location~\cite{chen10}.

We can solve ILPs using branch-and-bound as follows. We begin by finding a lower bound on the optimal solution to the ILP, by relaxing it to a standard linear program (LP) and solving the LP; that is, removing the constraint $x \in \Z^n$. This corresponds to the {\tt cost} function. If the solution $s$ is integer-valued, we are done, as it corresponds to a valid solution to the ILP. Otherwise, consider an index $i$ such that the found solution value $s_i$ is not an integer. To implement branching, we consider the two LPs formed by introducing the constraints $x_i \le \lfloor s_i \rfloor$, $x_i \ge \lceil s_i \rceil$. At least one of these must have the same optimal solution as the original ILP. We then repeat with these new LPs. An appealing aspect of this method is that the solution to the relaxation simultaneously tells us a lower bound on the cost, and a good variable to branch on.

The sequences of additional constraints specify subsets of potential solutions to the overall ILP. The {\tt branch} and {\tt cost} functions take this sequence as input and solve the resulting LP, to make a decision about which variable to branch on next, and compute a lower bound on cost, respectively. The complexity of the LP-solving step is polynomial in the input size, so the primary contribution to the overall runtime will in general be the exponential scaling in terms of the number of branching steps. A standard classical method could be used (e.g.\ the simplex algorithm), or one of the recently developed quantum algorithms for linear programming~\cite{brandao17,apeldoorn17,apeldoorn18,kerenidis18,brandao19}.

A particularly simple and elegant special case of this approach is the knapsack problem. Here we are given a list of $n$ items, each with weights $w_i$ and values $v_i$, and an overall weight upper bound $W$. We seek to find a subset $S$ of the items that maximises $\sum_{i \in S} v_i$, given that $\sum_{i \in S} w_i \le W$. We can write this as an integer linear program as follows:
\begin{align*}
\text{maximise } & \sum_{i=1}^n v_i x_i \\
\text{subject to } & \sum_{i=1}^n w_i x_i \le W,\\
& x_i \in \{0,1\} \text{ for all }i.
\end{align*}
Each variable $x_i$ corresponds to whether the $i$'th item is included in the knapsack. Then the LP relaxation is simply to replace the constraint $x_i \in \{0,1\}$ with the constraint $0 \le x_i \le 1$ for all $i$. This is equivalent to allowing fractional amounts of each item to be included.

The branch-and-bound approach to solving ILPs can immediately also be applied to the generalisation to Mixed Integer Linear Programming, where only certain variables are constrained. Now we only branch on those variables which are forced to be integers. One can also apply it to ``branch and cut'' algorithms. In this approach, when the LP relaxation returns a non-integer-valued solution, one may also add a new constraint (hyperplane) which separates that solution from all integer-valued feasible solutions.


\section{Analysis of Algorithm \ref{alg:qbb}}
\label{app:qbb}

In this appendix we prove the correctness and the claimed runtime bound of Algorithm \ref{alg:qbb}.

\begin{thm}
\label{thm:qbb}
Let $c_{\min}$ be the minimal cost of a valid solution, and let $T_{\min}$ be the size of the truncated tree with cost bound $c_{\min}$. (If there is no solution, $c_{\min} = \infty$, and $T_{\min}$ is the size of the whole tree.) Algorithm \ref{alg:qbb} uses
\begin{multline*}
O\bigg( \sqrt{T_{\min}d} \log c_{\max} \log\left(\frac{d \log c_{\max}}{\epsilon}\right)\\
\times \left( \log \left(\frac{d \log c_{\max}}{\epsilon}\right) + d \log d\right) \bigg)
 \end{multline*}
oracle calls, and except with failure probability at most $\epsilon$, returns a solution with minimal cost, if one exists, and otherwise ``no solution''.
\end{thm}

\begin{proof}
We first show that the algorithm succeeds with probability at least $1-\epsilon$. The loop executes at most $O(d)$ times, so each of Count and Search is used at most $O(d \log c_{\max})$ times. By a union bound, it is sufficient to pick $\epsilon' = O(\epsilon / (d \log c_{\max}))$ to ensure that all the uses of Count and Search succeed, except with total probability at most $\epsilon$. So we henceforth assume that Count and Search do always succeed.

If this is the case, we first observe that the algorithm always correctly outputs a minimal-cost solution, if one exists, or otherwise ``no solution''. This is because at the final iteration (when $T > T_{\max}/2$), if no solution has previously been found then Search will explore the entire tree and find a solution if one exists. To see that it outputs a {\em minimal-cost} solution, note that the binary search on $c$ using Search is over the range $[c_{\text{old}},c_{\text{new}}]$, and $c_{\text{old}}$ is no larger than the largest value of $c_{\text{new}}$ previously computed, so any solution with cost smaller than $c_{\text{old}}$ would have been found in a previous iteration.

It remains to prove the runtime bound. Let $T_c$ denote the size of the truncated tree with cost bound $c$ (so $T_{\min} = T_{c_{\min}}$). The first binary search (in part 2a) executes Count$_c$ $O(\log c_{\max})$ times, each iteration using $O(\sqrt{Td} \log^2(1/\epsilon'))$ queries; and the second binary search executes Search$_{c_{\text{new}}}$ $O(\log c_{\max})$ times, where each iteration uses $O(\sqrt{T_{c_{\text{new}}}} d^{3/2} \log d \log(1/\epsilon'))$ queries.
At each iteration of the loop, after the binary search using Count$_c$, $T_{c_{\text{new}}} \le 3T/2$ by correctness of the quantum tree size estimation algorithm.
Further, at the first iteration when $T \ge 3 T_{\min} / 2$ (if such an iteration occurs), for all $c \le c_{\min}$, Count$_c$ does not return ``contains more than $T$ nodes''. This implies that $c_{\text{new}} \ge c_{\min}$, because as the binary search terminated at cost $c_{\text{new}}$, Count$_{c_{\text{new}}+1}$ must have returned ``contains more than $T$ nodes''. Note that this holds even though Count$_c$ can return an arbitrary outcome when $T_c \in (2T/3, 3T/2]$.

Therefore, at this iteration the tree truncated at cost $c_{\text{new}}$ contains a minimal-cost solution, which will be found by the binary search on $c$ using Search$_c$, and the algorithm will terminate. On the other hand, if there is no iteration such that $T \ge 3 T_{\min} / 2$, we must have $T_{\min} > 2 T_{\max} / 3$. Combining these two claims, we have $T \le 3 T_{\min}$ throughout the algorithm.
The loop over exponentially increasing values of $T$ does not affect the overall complexity bound, so the overall complexity is
\begin{multline*}
O\bigg( \sqrt{T_{\min}d} \log c_{\max} \log\left(\frac{d \log c_{\max}}{\epsilon}\right)\\
\times \left( \log \left(\frac{d \log c_{\max}}{\epsilon}\right) + d \log d\right) \bigg)
\end{multline*}
queries, as claimed.
\end{proof}

We remark that it seems that, in general, Algorithm \ref{alg:qbb} could not be replaced with simply using the Search subroutine with exponentially increasing values of the cost parameter (an approach taken in~\cite{moylett17} for the special case of accelerating backtracking algorithms for the travelling salesman problem). This is because increasing the cost at which the tree is truncated by a constant factor could increase the size of the truncated tree substantially beyond $T_{\min}$.


\section{Truncated tree size bound for Sherrington-Kirkpatrick model}
\label{app:sktree}

In this appendix, let $T_A$ denote the size of the tree corresponding to the classical branch-and-bound algorithm applied to find the ground-state energy of an Ising Hamiltonian $H$ corresponding to an $n \times n$ matrix $A$, using the bounding function $\text{Bound}_A$ described in the main text, where the tree is truncated at the optimal value $c_{\min}$.

We will prove the following result:

\begin{thm}
\label{thm:skbound}
Let $A$ be an $n \times n$ matrix corresponding to a Sherrington-Kirkpatrick model instance on $n$ spins. For all sufficiently large $n$,
\[ \Pr_A[T_A \ge 2^{0.451n}] \le 0.01. \] 
\end{thm}

The dominant term in the quantum complexity is the square root of the classical complexity, which is determined by $T_A$, so Theorem \ref{thm:skbound} implies that the quantum branch-and-bound algorithm has an $O(2^{0.226n})$ running time on 99\% of Sherrington-Kirkpatrick model instances.

In order to prove Theorem \ref{thm:skbound}, we will need two technical lemmas, proven in Appendix \ref{app:lipschitz}.

\begin{lem}
\label{lem:lipschitz}
Let $x_1,\dots,x_N \sim N(0,1)$. Let $f:\R^N \rightarrow \R$ be continuous and 1-Lipschitz in each coordinate separately, i.e.\ $|f(x_1,\dots,x_N)-f(x_1,\dots,x_{i-1},x'_i,x_{i+1},\dots,x_N)| \le |x_i-x'_i|$ for all $i \in [N]$. Then
%
\beas \Pr[ f(x) \ge \E_x[f(x)] + t] \le e^{-t^2/(2N)},\\
\Pr[ f(x) \le \E_x[f(x)] - t] \le e^{-t^2/(2N)}. \eeas
\end{lem}

Lemma \ref{lem:lipschitz} was shown in~\cite{lalley13} but with an incorrect constant. We will also need a bound on the expectation $ \E_A[\min_{z \in \{\pm1\}^n} z^T A z]$. A precise value for this is known as $n \rightarrow \infty$, but we will need a bound that holds for arbitrary $n$:

\begin{lem}
\label{lem:expectation}
Let $A$ be an $n \times n$ matrix corresponding to a Sherrington-Kirkpatrick model instance on $n$ spins. For all $n$, $\E_A[\min_{z \in \{\pm1\}^n} z^T A z] \ge - 0.601\sqrt{n} - 0.833 n^{3/2}$.
\end{lem}

We are now able to prove Theorem \ref{thm:skbound}. The basic strategy is to upper-bound the expected value of $T_A$, using that (by linearity of expectation) this can be expressed as a sum over all bit-strings $x \in \{\pm1\}^\ell$, $0 \le \ell \le n$, of the probability that the node corresponding to $x$ is contained within the truncated branch-and-bound tree. These bit-strings $x$ are precisely those such that Bound$_A(x) \le \min_{z \in \{\pm1\}^n} z^T A z$, and a tail bound can be used to upper-bound the probability that this event occurs.

There are two technical difficulties which need to be handled. First, this approach does not give a good upper bound in the case where $\min_{z \in \{\pm1\}^n} z^T A z$ is high, which can occur with non-negligible probability, leading to $\E_A[T_A]$ becoming large. We therefore handle this case separately and show that it occurs with low probability. Next, to find a tail bound on Bound$_A(x)$, we need to compute expressions of the form $\E_A[\min_{z \in \{\pm1\}^{n-\ell}} z^T A z]$; although a limiting form for this is known~\cite{parisi80,talagrand06,schmidt08}, we will additionally need relatively tight bounds in the case $\ell \approx n$. We therefore split into cases $\ell \le 0.9n$ (where $(n-\ell) \rightarrow \infty$ and we can use the precise limiting result) and $\ell > 0.9n$ (where we use Lemma \ref{lem:expectation}).

\begin{widetext}
\begin{proof}[Proof of Theorem \ref{thm:skbound}]
Write $\mu = \E_A[\min_{z \in \{\pm1\}^n} z^T A z]$, and let $\gamma > 0$ be an arbitrary value to be determined. We will upper-bound the probability that $T_A \ge B$ for some $B$ as follows, where we use the notation $[Y]$ for the indicator random variable which evaluates to 1 if $Y$ is true, and 0 if $Y$ is false:
\beas
\Pr_A[T_A \ge B] &=& \Pr_A[T_A \ge B \wedge \min_z z^T A z \le \mu + \gamma n^{3/2}] + \Pr_A[T_A \ge B \wedge \min_z z^T A z > \mu + \gamma n^{3/2}] \\
&\le& \Pr_A\left[T_A[\min_z z^T A z \le \mu + \gamma n^{3/2}] \ge B \right] + \Pr_A[\min_z z^T A z > \mu + \gamma n^{3/2}]\\
&\le& \frac{1}{B} \E_A\left[ T_A [\min_z z^T A z \le \mu + \gamma n^{3/2}] \right] + \Pr_A[\min_z z^T A z > \mu + \gamma n^{3/2}]\\
&=& \frac{1}{B} \sum_{\ell=0}^n \sum_{x \in \{\pm1\}^\ell} \E_A\left[ [\text{Bound}_A(x) \le  \min_z z^T A z] [\min_z z^T A z \le \mu + \gamma n^{3/2}] \right] + \Pr_A[\min_z z^T A z > \mu + \gamma n^{3/2}]\\
&\le& \frac{1}{B} \sum_{\ell=0}^n \sum_{x \in \{\pm1\}^\ell} \Pr_A\left[ \text{Bound}_A(x) \le \mu + \gamma n^{3/2} \right] + \Pr_A[\min_z z^T A z > \mu + \gamma n^{3/2}]\\
&\le& \frac{n+1}{B} \max_{\ell, x \in \{\pm1\}^\ell} 2^\ell \Pr_A\left[ \text{Bound}_A(x) \le \mu + \gamma n^{3/2} \right] + \Pr_A[\min_z z^T A z > \mu + \gamma n^{3/2}]\\
\eeas
where the second inequality is Markov's inequality and we use linearity of expectation in the second equality.

To upper-bound the last term, we use Lemma \ref{lem:lipschitz}. We first observe that $f(A) = \min_z z^T A z$ is 1-Lipschitz in each variable, as if we modify $A$ to produce $A'$ by changing $a_{pq}$ to $a'_{pq}$ for some pair $p<q$,
\be \label{eq:1lipschitz} \min_{z \in \{\pm1\}^n} \sum_{i < j} a'_{ij} z_i z_j =  \min_{z \in \{\pm1\}^n} \left((a'_{pq} - a_{pq})z_p z_q + \sum_{i < j} a_{ij} z_i z_j \right) \ge -|a'_{pq} - a_{pq}| + \min_{z \in \{\pm1\}^n} \sum_{i < j} a_{ij} z_i z_j, \ee
and by a similar argument $\min_{z \in \{\pm1\}^n} \sum_{i < j} a'_{ij} z_i z_j \le |a'_{pq} - a_{pq}| + \min_{z \in \{\pm1\}^n} \sum_{i < j} a_{ij} z_i z_j$. So Lemma \ref{lem:lipschitz} implies that
\[ \Pr_A[\min_z z^T A z > \mu + \gamma n^{3/2}] \le e^{-(\gamma n^{3/2})^2/(2 \binom{n}{2})} \le e^{-\gamma^2 n}. \]
For this to be upper-bounded by a small constant (e.g.\ 0.005) we can take $\gamma = O(1/\sqrt{n})$.

We next upper-bound the first term by bounding $\Pr_A\left[ \text{Bound}_A(x) \le \mu + \gamma n^{3/2} \right]$. We only need to consider $\ell \ge 0.4n$ in the maximisation, because when $\ell \le 0.4n$, trivially upper-bounding this probability by 1 already gives a sufficiently strong bound. Recall that
\[ \text{Bound}_A(x) = \sum_{1 \le i < j \le \ell} a_{ij} x_i x_j - \sum_{j =\ell+1}^n \left| \sum_{i=1}^{\ell} a_{ij} x_i \right| +  \min_z \sum_{\ell+1 \le i < j \le n} a_{ij} z_i z_j. \]
The function $f_x(A) = \text{Bound}_A(x)$ is 1-Lipschitz in each variable by a similar argument to (\ref{eq:1lipschitz}). Thus, for any $\eta \ge 0$,
\be \label{eq:prbound} \Pr_A\left[\text{Bound}_A(x) \le \E_A[\text{Bound}_A(x)] - \eta n^{3/2} \right] \le e^{-\eta^2 n}. \ee
First assume that $0.4n \le \ell \le 0.9n$, so $(n-\ell) \rightarrow \infty$ as $n\rightarrow \infty$. For any $x \in \{\pm1\}^\ell$, we have
\bea
\E_A[\text{Bound}_A(x)] &=& \sum_{1 \le i < j \le \ell} \E_A[a_{ij}] x_i x_j - \sum_{j =\ell+1}^n \E_A\left[\left| \sum_{i=1}^{\ell} a_{ij} x_i \right|\right] + \E_A\left[ \min_z \sum_{\ell+1 \le i < j \le n} a_{ij} z_i z_j \right] \\
&=& - (n-\ell) \E_A\left[\left| \sum_{i=1}^{\ell} a_{ij} \right|\right] + \E_A\left[ \min_z \sum_{\ell+1 \le i < j \le n} a_{ij} z_i z_j \right] \\
\label{eq:g1alpha} &=& - (n-\ell) \sqrt{\frac{2}{\pi}} \sqrt{\ell} - (0.763\dots + o(1))(n-\ell)^{3/2},
\eea
where we use linearity of expectation to obtain the first expression, that $\sum_{i=1}^{\ell} a_{ij} \sim \sqrt{\ell} N(0,1)$, and the known limiting result $\E_A\left[ \min_z \sum_{\ell+1 \le i < j \le n} a_{ij} z_i z_j \right] = (-0.763167\dots + o(1))(n-\ell)^{3/2}$~\cite{parisi80,talagrand06,schmidt08}.

Writing $\alpha = \ell/n$, we have that
\[ \E_A[\text{Bound}_A(x)] = (- (1-\alpha)\sqrt{\alpha} \sqrt{\frac{2}{\pi}} - (0.763\dots + o(1))(1-\alpha)^{3/2})n^{3/2} =: g_1(\alpha) n^{3/2}. \]
On the other hand, for $\ell \ge 0.9n$, we follow a similar argument but apply the nonasymptotic result of Lemma \ref{lem:expectation} to bound $\E_A\left[ \min_z \sum_{\ell+1 \le i < j \le n} a_{ij} z_i z_j \right]$, which implies that
\bea
\E_A[\text{Bound}_A(x)] &\ge& - (n-\ell) \sqrt{\frac{2}{\pi}} \sqrt{\ell} - 0.601\sqrt{n-\ell} - 0.833 (n-\ell)^{3/2}\\
&\ge& - (n-\ell) \sqrt{\frac{2}{\pi}} \sqrt{\ell} - 1.434 (n-\ell)^{3/2}\\
\label{eq:g2alpha} &=&  (- (1-\alpha)\sqrt{\alpha} \sqrt{\frac{2}{\pi}} - 1.434 (1-\alpha)^{3/2})n^{3/2} =: g_2(\alpha) n^{3/2}.
\eea
In either case, we have
\[ \Pr_A\left[\text{Bound}_A(x)  \le \mu + \gamma n^{3/2}\right] = \Pr_A\left[\text{Bound}_A(x) - \E_A[\text{Bound}_A(x)] \le \mu + \gamma n^{3/2} - \E_A[\text{Bound}_A(x)] \right]. \]
By (\ref{eq:prbound}), using $\mu = (-0.763167\dots + o(1))n^{3/2}$ and observing (see Figure \ref{fig:galpha}) that $\E_A[\text{Bound}_A(x)] > \mu + \gamma n^{3/2}$ for sufficiently large $n$, so the right-hand side is negative as required, we have
\[ \Pr_A\left[\text{Bound}_A(x)  \le \mu + \gamma n^{3/2}\right] \le \begin{cases} e^{-(g_1(\alpha) + (0.763\dots + o(1)) - \gamma)^2 n} & \text{if } 0.4 \le \alpha \le 0.9 \\ e^{-(g_2(\alpha) + (0.763\dots + o(1)) - \gamma)^2 n} & \text{if } \alpha \ge 0.9 \end{cases}. \]
So 
\beas
&& \!\!\!\!\!\!\!\!\!\!\!\!\!\!\!\!\!\!\!\! \!\!\!\!  \max_{\ell \ge 0.4n, x \in \{\pm1\}^\ell} 2^\ell \Pr_A\left[ \text{Bound}_A(x) \le \mu + \gamma n^{3/2} \right]\\
 &\le& \max \{ \max_{\alpha \in [0.4,0.9]} 2^ {\alpha n} e^{-(g_1(\alpha) + 0.763\dots + o(1) )^2 n}, \max_{\alpha \in [0.9,1]} 2^ {\alpha n} e^{-(g_2(\alpha) + 0.763\dots + o(1) )^2 n} \} \\
&=& \max \{ \max_{\alpha \in [0.4,0.9]} 2^{n(\alpha - (g_1(\alpha) + 0.763\dots + o(1))^2/\ln 2)}, \max_{\alpha \in [0.9,1]} 2^{n(\alpha - (g_2(\alpha) + 0.763\dots + o(1))^2/\ln 2)} \}
\eeas
observing that $\gamma = o(1)$.
It remains to determine upper bounds on the functions
\be \label{eq:halpha} \alpha - \frac{(g_1(\alpha) + 0.763\dots + o(1))^2}{\ln 2} = \alpha - \frac{(- (1-\alpha)\sqrt{\alpha} \sqrt{\frac{2}{\pi}} + 0.763\dots(1-(1-\alpha)^{3/2}))^2}{\ln 2} + o(1) =: h_1(\alpha) + o(1), \ee
\be \label{eq:halpha2} \alpha - \frac{(g_2(\alpha) + 0.763\dots + o(1))^2}{\ln 2} = \alpha - \frac{(- (1-\alpha)\sqrt{\alpha} \sqrt{\frac{2}{\pi}} + 0.763\dots - 1.434(1-\alpha)^{3/2})^2}{\ln 2} + o(1) =: h_2(\alpha) + o(1). \ee
This can easily be achieved numerically, giving (see Figure \ref{fig:halpha}) the result that $h_1(\alpha) < 0.45003$ for $0 \le \alpha \le 1$ and $h_2(\alpha) < 0.45003$ for $\alpha \ge 0.9$. Hence
\[ \Pr_A[T_A \ge B] \le \frac{(n+1) 2^{(0.45003 + o(1))n}}{B} + 0.005, \]
and to upper-bound the first term by 0.005, for sufficiently large $n$ one can take $B = 2^{0.451n}$. This completes the proof.
\end{proof}
\end{widetext}

\begin{figure}
\includegraphics[width=0.5\textwidth]{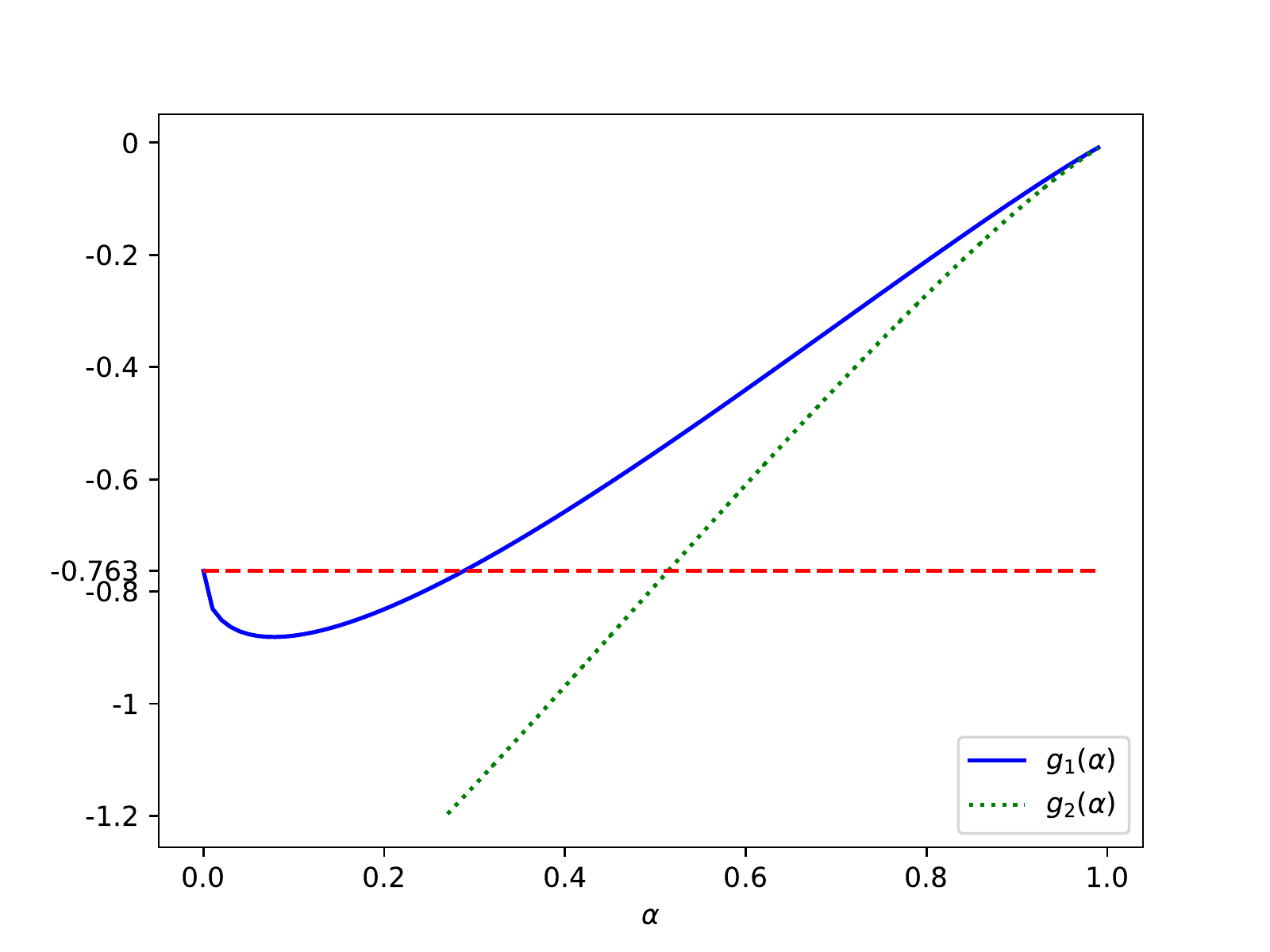}
\caption{The functions $g_1(\alpha)$, $g_2(\alpha)$ defined in (\ref{eq:g1alpha}), (\ref{eq:g2alpha}). $g_1(\alpha) \ge -0.763$ for all $\alpha \ge 0.4$, while $g_2(\alpha) \ge -0.763$ for all $\alpha \ge 0.9$.}
\label{fig:galpha}
\end{figure}

\begin{figure}
\includegraphics[width=0.5\textwidth]{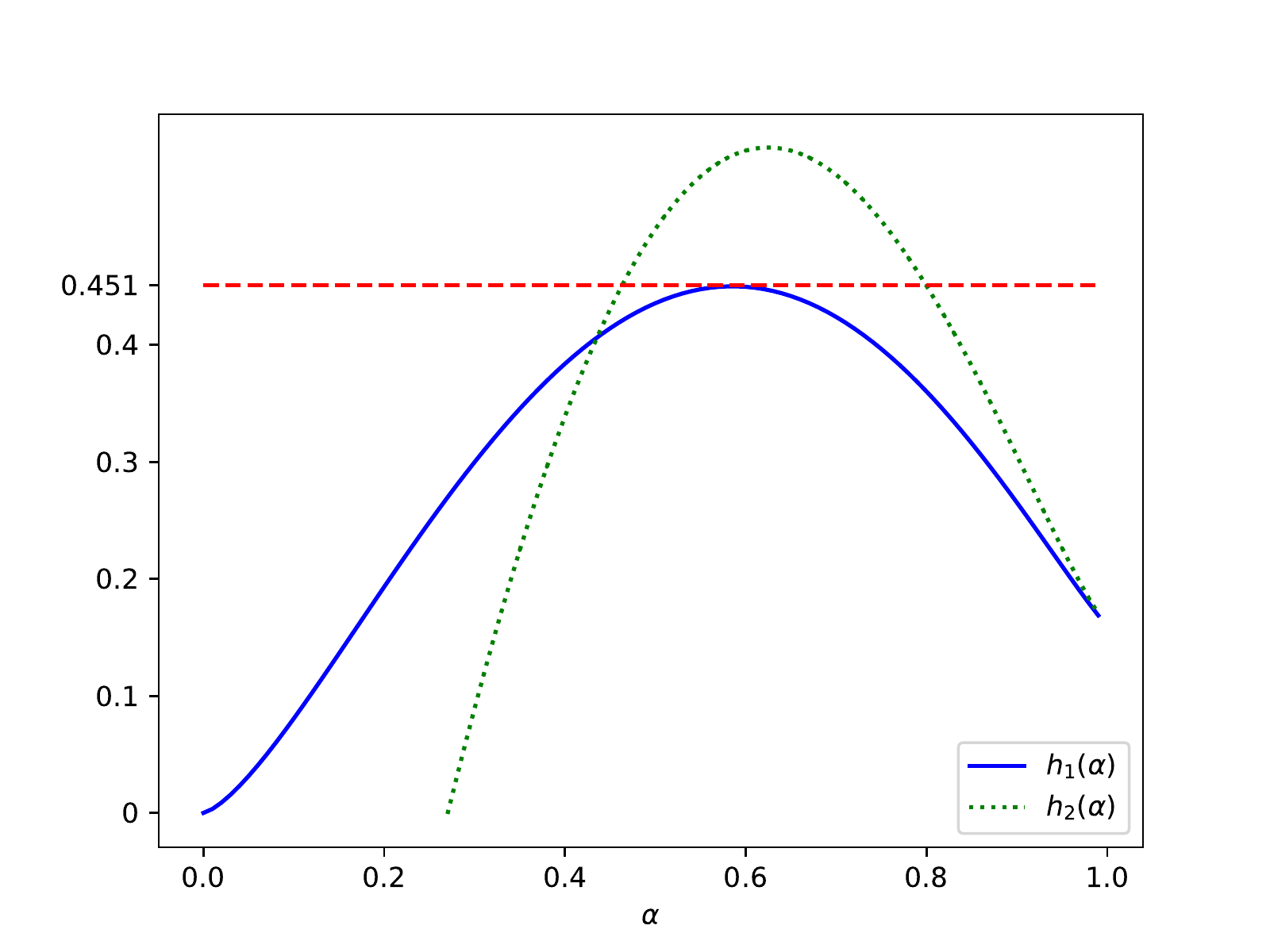}
\caption{The functions $h_1(\alpha)$, $h_2(\alpha)$ defined in (\ref{eq:halpha}), (\ref{eq:halpha2}). $h_1(\alpha) < 0.45003$ for all $\alpha$, while $h_2(\alpha) < 0.45003$ for all $\alpha \ge 0.9$.}
\label{fig:halpha}
\end{figure}


\section{Proofs of technical lemmas}
\label{app:lipschitz}

In this appendix we prove Lemmas \ref{lem:lipschitz} and \ref{lem:expectation}. We say that $f(x_1,\dots,x_N)$ satisfies the bounded differences condition with constants $d_i$, $i \in [N]$, if $|f(x) - f(x')| \le d_i$ whenever $x$ and $x'$ differ only in the $i$'th coordinate.

\begin{lem}[McDiarmid's inequality or method of bounded differences~{\cite[Corollary 5.2]{dubhashi09}}]
\label{lem:bdmethod}
If $f(x_1,\dots,x_N)$ satisfies the bounded differences condition with constants $d_i$, and $x_1,\dots,x_N$ are independent random variables, then
\beas
\Pr[ f(x) \ge \E_x[f(x)] + t] \le e^{-2t^2/d},\\
\Pr[ f(x) \le \E_x[f(x)] - t] \le e^{-2t^2/d},
\eeas
where $d = \sum_{i=1}^N d_i^2$.
\end{lem}

\begin{replem}{lem:lipschitz}
Let $x_1,\dots,x_N \sim N(0,1)$. Let $f:\R^N \rightarrow \R$ be continuous and 1-Lipschitz in each coordinate separately, i.e.\ $|f(x_1,\dots,x_N)-f(x_1,\dots,x_{i-1},x'_i,x_{i+1},\dots,x_N)| \le |x_i-x'_i|$ for all $i \in [N]$. Then
%
\beas \Pr[ f(x) \ge \E_x[f(x)] + t] \le e^{-t^2/(2N)},\\
\Pr[ f(x) \le \E_x[f(x)] - t] \le e^{-t^2/(2N)}.
\eeas
\end{replem}

\begin{proof}
For $i \in [N]$, $j \in [M]$, let $y_i^j$ be a Rademacher random variable, taking values $\pm 1$ with equal probability. Then define the sequence $x(y)$ by $x(y)_i = \frac{1}{\sqrt{M}} \sum_{j=1}^M y_i^j$. Let $g:\{\pm1\}^{MN} \rightarrow \R$ be defined by setting $g(y) = f(x(y))$. Then changing one entry of $y$ can change $x(y)_i$ by at most $2/\sqrt{M}$, so we can apply Lemma \ref{lem:bdmethod} with $d_i = 2/\sqrt{M}$ to obtain
\beas
\Pr[ f(x(y)) \ge \E_y[f(x(y))] + t] \le e^{-t^2 / (2N)},\\
\Pr[ f(x(y)) \le \E_y[f(x(y))] - t] \le e^{-t^2/ (2N)}.
\eeas
As $M \rightarrow \infty$, the distribution of $x(y)_i$ approaches a standard normal distribution for all $i$. The lemma follows.
\end{proof}

\begin{replem}{lem:expectation}
Let $A$ be an $n \times n$ matrix corresponding to a Sherrington-Kirkpatrick model instance on $n$ spins. For all $n$, $\E_A[\min_{z \in \{\pm1\}^n} z^T A z] \ge - 0.601\sqrt{n} - 0.833 n^{3/2}$.
\end{replem}

\begin{proof}[Proof of Lemma \ref{lem:expectation}]
We start with
\[ \E_A\left[\min_{z \in \{\pm1\}^n} z^T A z\right] = -\int_{-\infty}^0 \Pr\left[\min_{z \in \{\pm1\}^n} z^T A z \le t\right] dt, \]
valid as $\min_{z \in \{\pm1\}^n} z^T A z$ is non-positive. Next, for any $t \le 0$ we have
\[ \Pr\left[\min_{z \in \{\pm1\}^n} z^T A z \le t\right] \le 2^n \Pr\left[ \sum_{i < j} a_{ij} \le t \right] \]
using a union bound over $z$ and symmetry of the distribution of $a_{ij}$. By a tail bound on the normal distribution, we have
\[ \Pr\left[ \sum_{i < j} a_{ij} \le t \right] \le e^{-t^2/(2 \binom{n}{2})} \le e^{-t^2/n^2} \]
for all $t \le 0$. So
\beas
&& \E_A\left[\min_{z \in \{\pm1\}^n} z^T A z\right]\\
&\ge& -\int_{-\infty}^0 \min\{1, 2^n e^{-t^2/n^2} \} dt\\
&=& -\int_{-\infty}^{-n^{3/2}\sqrt{\ln 2} } 2^n e^{-t^2/n^2} dt - \int_{-n^{3/2}\sqrt{\ln 2} }^0 1 dt\\
&=& -n 2^{n-1/2} \int_{-\infty}^{-\sqrt{2n \ln 2} } e^{-t^2/2} dt - \sqrt{\ln 2} n^{3/2}\\
&\ge& - \frac{\sqrt{n}}{2 \sqrt{\ln 2}}  - \sqrt{\ln 2} n^{3/2}\\
&=& - 0.600561\dots \sqrt{n} - 0.832555\dots n^{3/2},
\eeas
where we use the bound $\int_a^\infty e^{-x^2/2} dx \le \frac{1}{a} e^{-a^2/2}$ for any $a>0$ in the second inequality.
\end{proof}


\section{Classical numerical branch-and-bound results}
\label{app:numerics}

We implemented the classical branch-and-bound algorithm described in the main text, with cost function Bound$_A$, using a simple depth-first search procedure within the branch-and-bound tree which backtracks on nodes corresponding to partial solutions with an energy bound worse than the lowest energy seen thus far. For an S-K model instance described by a matrix $A$, this gives an upper bound on the size $T_A$ of an optimally truncated branch-and-bound tree (equivalently, on the runtime of the best-first search algorithm applied to find the ground state energy, with cost function Bound$_A$).

This algorithm enabled instances on more than 50 spins to be solved within minutes on a standard laptop computer. We then carried out a least-squares fit on the log of the number of nodes explored, omitting small $n$, to estimate the scaling of the algorithm with $n$. Note that, due to finite-size effects, this may not be accurate for large $n$; however, it gives an indication of tree size scaling. The median normalised ground state energy found for the larger values of $n$ (e.g.\ $\approx -0.71$ for $n=50$) seems to approach the limiting value $-0.763167\dots$ relatively slowly. These results are consistent with heuristic finite-size results reported in~\cite{boettcher05} and were validated using exhaustive search for small $n$.

\begin{figure}
\includegraphics[width=0.5\textwidth]{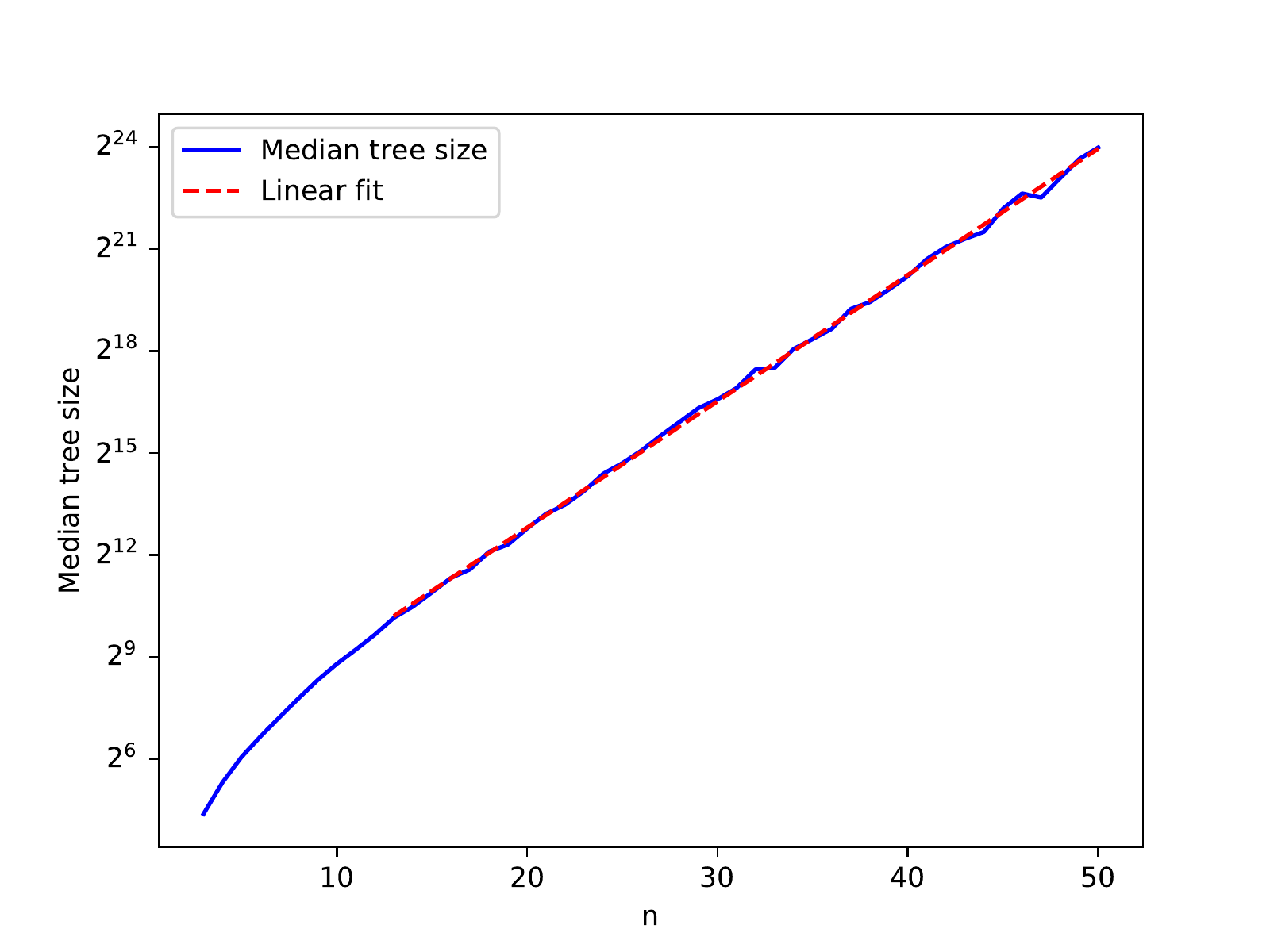}
\caption{Median tree size explored by classical branch-and-bound algorithm with depth-first strategy. 99 random instances generated for each $n$. Fit is line $y = 2^{0.371n + 5.380}$.}
\label{fig:treesizes}
\end{figure}

\begin{figure}
\includegraphics[width=0.5\textwidth]{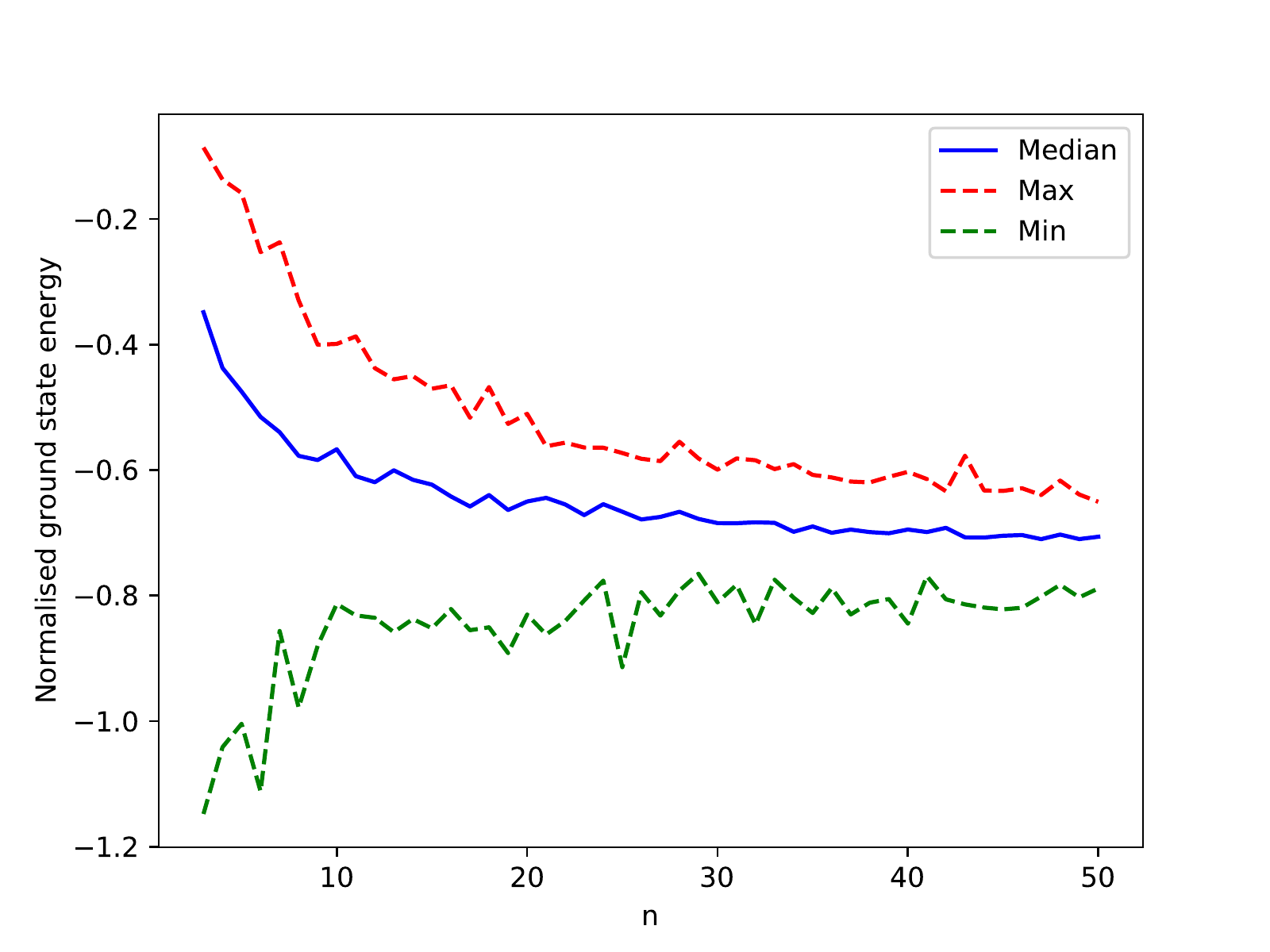}
\caption{Normalised ground state energy $E_{\min} n^{-3/2}$ of instances of the S-K model. 99 random instances generated for each $n$.}
\label{fig:energies}
\end{figure}

\end{document}